\documentclass[11pt, onecolumn, anonymous]{IEEEtran}
\usepackage[letterpaper, margin=1in]{geometry}  
\usepackage{times}  
\usepackage{amsmath, amssymb, amsthm}  
\usepackage{graphicx}  
\usepackage{hyperref}  

\renewcommand{\baselinestretch}{1}
\usepackage{tikz}
\usepackage{algorithm}
\usepackage[noend]{algpseudocode}
\usepackage{subcaption}
\usepackage{lineno}
\usepackage[letterpaper, margin=1in]{geometry}

\begin{document}

\title{Finite Guarding of Weakly Visible Segments via Line Aspect Ratio in Simple Polygons}

\author{Arash~Vaezi, 
\thanks{Department of Computer Science, Institute for Research in Fundamental Sciences (IPM), e-mail: avaezi@sharif.edu; avaezi@ipm.ir}
}
\maketitle


\newtheorem{theorem}{Theorem}
\newtheorem{assumption}{Assumption}
\newtheorem{lemma}{Lemma}
\newtheorem{obs}{Observation}
\newtheorem{definition}{Definition}
\newtheorem{remark}{Remark}

\def\for{\mbox{\it for}}
\def\rf{\mbox{\it rf}}
\def\AS{\mbox{\it AS}}
\def\VP{\mbox{\it VP}}
\def\Vp{\mbox{\it Vp}}
\def\VL{\mbox{\it VL}}
\def\scVL{\mbox{\it scVL}}
\def\VA{\mbox{\it VA}}
\def\TP{\mbox{\it TP}}
\def\TI{\mbox{\it TI}}
\def\NP{\mbox{\it NP}}
\def\UC{\mbox{\it UC}}
\def\GR{\mbox{\it GR}}
\def\PR{\mbox{\it PR}}
\def\UB{\mbox{\it UB}}
\def\WVP{\mbox{\it WVP}}
\def\SVL{\mbox{\it SVL}}
\def\SVP{\mbox{\it SVP}}
\def\LBV{\mbox{\it LBV}}
\def\RBV{\mbox{\it RBV}}
\def\SGR{\mbox{\it SGR}}
\def\CVP{\mbox{\it CVP}}
\def\OPT{\mbox{\it OPT}}
\def\SCR{\mbox{\it SCR}}
\def\SSCR{\mbox{\it SSCR}}
\def\SR{\mbox{\it SR}}
\def\UWVP{\mbox{\it UWVP}}
\def\UVL{\mbox{\it UVL}}
\def\PVL{\mbox{\it PVL}}
\def\PVP{\mbox{\it PVP}}
\def\UWVP{\mbox{\it UWVP}}
\def\EdP{\mbox{\it EdP}}
\def\EdI{\mbox{\it EdI}}
\def\Findtsr{\mbox{\it Findtsr}}
\def\Decomposescr{\mbox{\it Decomposescr}}
\def\inPL{\mbox{\it inPL}}
\def\seg#1{\overline{#1}}
\def\P{\mathcal P}
\def\T{\mathcal T}
\def\S{\mathcal S}
\def\F{\mathcal F}
\def\M{\mathcal M}
\def\G{\mathcal G}
\def\D{\mathcal D}

\begin{abstract}

We address the problem of covering a target segment $\overline{uv}$ using a finite set of guards $\mathcal{S}$ placed on a source segment $\overline{xy}$ within a simple polygon $\mathcal{P}$, assuming weak visibility between the target and source. Without geometric constraints, $\mathcal{S}$ may be infinite, as shown by prior hardness results. To overcome this, we introduce the {\it line aspect ratio} (AR), defined as the ratio of the \emph{long width} (LW) to the \emph{short width} (SW) of $\mathcal{P}$. These widths are determined by parallel lines tangent to convex vertices outside $\mathcal{P}$ (LW) and reflex vertices inside $\mathcal{P}$ (SW), respectively.

Under the assumption that AR is constant or polynomial in $n$ (the polygon's complexity), we prove that a finite guard set $\mathcal{S}$ always exists, with size bounded by $\mathcal{O}(\text{AR})$. This AR-based framework generalizes some previous assumptions, encompassing a broader class of polygons.

Our result establishes a framework guaranteeing finite solutions for segment guarding under practical and intuitive geometric constraints.

\end{abstract}


\newpage

\IEEEpeerreviewmaketitle

\section{Introduction}
\label{sec:introduction}
Let $\mathcal{P}$ be a simple polygon, and let $\text{int}(\mathcal{P})$ denote its interior. Two points $x$ and $y$ in $\mathcal{P}$ are \emph{visible} to each other if and only if the relatively open line segment $\overline{xy}$ lies entirely within $\text{int}(\mathcal{P})$. The \emph{visibility polygon} of a point $q \in \mathcal{P}$, denoted by $\text{VP}(q)$, is the set of all points in $\mathcal{P}$ visible to $q$. 

The \emph{weak visibility polygon} of a line segment $\overline{pq}$, denoted $\text{WVP}(\overline{pq})$, is the maximal sub-polygon of $\mathcal{P}$ visible to at least one interior point of $\overline{pq}$. A polygon $\mathcal{P}$ is \emph{completely visible} from a segment $\overline{pq}$ (denoted $\text{CVP}$) if every point $z \in \mathcal{P}$ is visible from every point $w \in \overline{pq}$. Algorithms to compute $\text{WVP}$ and $\text{CVP}$ in linear time are known \cite{avis,2}.

\section{Literature Review}
\label{sec:lit-review}

Avis et. al. \cite{avis1981optimal} developed a linear-time method for computing weak visibility polygons, enabling the characterization of regions visible from a given segment.
Similarly, Guibas et. al. \cite{guibas1987visibility} introduced trapezoidal decomposition methods to support efficient segment-to-segment visibility queries. These techniques are crucial for testing the visibility conditions in different scenarios. Complementing these, Hershberger et. al. \cite{hershberger1989optimal} proposed an optimal algorithm for visibility graph construction, effectively determining the visibility relationships between points on two segments.

The combinatorial conditions for mutual visibility between segments have been extensively studied. Sack et. al. \cite{sack1991weak} established necessary and sufficient conditions for weak mutual visibility, providing a duality-based framework that explains cases where partial visibility permits finite covering sets. Later, Ghodsi et. al. \cite{ghodsa2008weak} developed decision algorithms for testing weak visibility between disjoint segments, showing that determining whether $\overline{uv} \subseteq \text{WVP}(\overline{xy})$ is solvable in $\mathcal{O}(n \log n)$ time.

Guard placement optimization has been explored in related contexts. T{\'o}th \cite{toth2010art} demonstrated that $\lfloor n/4 \rfloor$ edge guards suffice to cover simple polygons. King et. al. \cite{king2004guarding} investigated mobile guards along segments. Additionally, Bhattacharya et. al. \cite{bhattacharya2015approximability} provided an $\mathcal{O}(\log n)$-approximation for guarding weakly visible polygons.

The study of visibility within simple polygons has advanced significantly, particularly with specialized cases like sliding cameras, reflections, and structured segment visibility. Biedl et al. explored sliding-camera guards, providing approximation algorithms and demonstrating the NP-hardness of sliding-camera coverage in polygons with holes \cite{biedl2017sliding}. Vaezi et. al. addressed reflection-extended visibility, where polygon edges act as mirrors, enabling previously invisible segments to become visible; their work covers weak, strong, and complete visibility settings \cite{vaeziwalcom,tcs}. Lee and Chwa \cite{lee1990chain} focused on chain visibility, investigating the visibility of polygonal chains and providing efficient algorithms for both convex and reflex chains. Recent research has introduced k-transmitters, extending visibility to cases where light rays may cross polygon boundaries multiple times \cite{algowin2023transmitters}. Furthermore, structured visibility profiles and efficient data structures for segment-to-segment queries have been studied extensively, offering solutions for visibility tracking and analysis in dynamic environments \cite{chen1998structured}.

The inherent difficulty of unrestricted guarding has motivated the introduction of geometric constraints. Bonnet et. al. \cite{bonnet2016approximating} proved the APX-hardness of guarding problems, highlighting the necessity of assumptions like their vision stability or our line aspect ratio condition. Notably, their results demonstrate that without such constraints, the set of guards $\mathcal{S}$ may be infinite—a key motivation for our approach. Unlike previous work relying on integer-coordinate assumptions, our framework accommodates real-coordinate polygons with exponential complexity, generalizing these results.

\subsection{Positioning of Our Contribution}

Existing research has laid a strong foundation in:
\begin{itemize}
    \item 
 Efficient computation of weak visibility polygons \cite{avis1981optimal}
   \item  Visibility testing between segments \cite{ghodsa2008weak}
   \item Hardness results for general guarding problems \cite{bonnet2016approximating}
\end{itemize}
Our work extends this body of knowledge by providing:

\begin{itemize}
\item A guarantee of finite guard sets $\mathcal{S}$ for segment coverage under the line aspect ratio assumption
\item  Generalization to polygons excluded by prior vision stability assumptions (Theorem~\ref{thm:AR-generalizes-BM})
\item  Explicit bounds on guard set size, $|\mathcal{S}| = \mathcal{O}(\text{AR})$ (Theorem~\ref{thm:finite-S})
\end{itemize}

\section{Problem Definition}
\label{sec:problem}
Consider two line segments: a \emph{target} segment $\overline{uv}$ and a \emph{source} segment $\overline{xy}$. The visibility between these segments may fall into one of three cases:
\begin{enumerate}
    \item $\overline{uv}$ and $\overline{xy}$ are completely visible ($\text{CVP}$).
    \item At least one point of $\overline{xy}$ or $\overline{uv}$ is invisible to the other segment.
    \item $\overline{uv}$ and $\overline{xy}$ are partially visible (i.e., every point on one segment is visible to at least one point on the other, but not necessarily all points). From now on we refer to this case as the target is \textit{weakly visible} from the source.
\end{enumerate}

This work focuses on the third case. We aim to find a finite and polynomial set $\mathcal{S}$ of points on the source segment $\overline{xy}$ such that:
\[
\overline{uv} \subseteq \bigcup_{s \in \mathcal{S}} \text{VP}(s)
\]
where $\text{VP}(s)$ denotes the visibility polygon of $s$. We assume $\overline{uv}$ is weakly visible from $\overline{xy}$.

\section{Assumptions}
\label{sec:assumptions}
To ensure a finite size for $\mathcal{S}$, certain assumptions about $\mathcal{P}$ are necessary. This section introduces these assumptions. Specifically, we present an algorithm for determining the points of $\mathcal{S}$ and demonstrate that, under Assumption \ref{assumme:AR1}, the algorithm yields a finite set $\mathcal{S}$ whose size is polynomial in $n$.

\begin{assumption}[Line Aspect Ratio (Our assumption)]
\label{assumme:AR1}
For a simple polygon $\mathcal{P}$, the \emph{long width} ($LW$) is the maximum distance between two parallel lines tangent to convex vertices of $\mathcal{P}$, on the outside of $\mathcal{P}$ without intersecting its interior. The \emph{short width} ($SW$) is the minimum distance between two such parallel line segments tangent to the reflex vertices in the interior of $\P$ and constrained by the polygon's boundary. The \emph{line aspect ratio} is:
\[
\text{AR}_{\text{line}} = \frac{LW}{SW}
\]
One may consider two cases:
\begin{itemize}
    \item \emph{Constant} line aspect ratio: $\text{AR}_{\text{line}} = O(1)$
    \item \emph{Polynomial} line aspect ratio: $\text{AR}_{\text{line}} = \text{poly}(n)$
\end{itemize}
where $n$ is the complexity of $\mathcal{P}$.
\end{assumption}

\begin{assumption}[Disk Aspect Ratio]
\label{assume:AR2}
For a simple polygon $\mathcal{P}$, the \emph{long diameter} ($LD$) is the diameter of the smallest enclosing circle tangent to the boundary. The \emph{short diameter} ($SD$) is the diameter of the largest inscribed circle tangent to the boundary. The \emph{disk aspect ratio} is:
$$
\text{AR}_{\text{disk}} = \frac{LD}{SD}
$$
One may consider:
\begin{itemize}
    \item \emph{Constant} disk aspect ratio: $\text{AR}_{\text{disk}} = O(1)$
    \item \emph{Polynomial} disk aspect ratio: $\text{AR}_{\text{disk}} = \text{poly}(n)$
\end{itemize}
\end{assumption}

\begin{definition}[General Position \cite{BonnetMiltzow2016}]
\label{def:general_position}
A simple polygon $\mathcal{P}$ satisfies the \emph{general position} assumption if no three distinct vertices are collinear.
\end{definition}

\begin{definition}[Vision Stability \cite{BonnetMiltzow2016}]
\label{def:vision_stability}
$\mathcal{P}$ satisfies \emph{vision stability} if there exists $\gamma = \gamma(\mathcal{P}) > 0$ such that for any $p, q \in \mathcal{P}$:

$$
\mu\left( \VP(p) \triangle \VP(q) \right) \leq \gamma \cdot \|p - q\|
$$
where:
\begin{itemize}
  \item $\VP(x)$ is the visibility polygon of $x$
  \item $\triangle$ denotes symmetric difference
  \item $\mu$ is the Lebesgue measure
  \item $\|\cdot\|$ is the Euclidean norm
\end{itemize}
\end{definition}

\begin{assumption}[Vision Stability]
\label{assume:vision-stability}
A simple polygon $\mathcal{P}$ satisfies:
\begin{enumerate}
  \item General position (Definition~\ref{def:general_position})
  \item Vision stability (Definition~\ref{def:vision_stability})
\end{enumerate}
\end{assumption}

\subsection{Connection to Art Gallery Problems}

Our defined problem is inspired by the Art Gallery Problem, which seeks the minimum number of guards (points) required to observe every interior point of a polygon $\mathcal{P}$ with $n$ vertices. Bonnet and Miltzow \cite{cite:EdouardTilmann} highlighted fundamental limitations regarding visibility coverage:

\textit{Lemma 4 of \cite{cite:EdouardTilmann}}: For any finite set $D$ within $\mathcal{P}$, there exists a point $x \in \mathcal{P}$ such that $x$ can see a point $p$ invisible to all $d \in D$ within distance 1 of $x$.

\textit{Lemma 5 of \cite{cite:EdouardTilmann}}: For any constant $c \in \mathbb{N}$, there exists a polygon $\mathcal{P}_c$ where, for any finite set $D$, some $x \in \mathcal{P}_c$ has a visibility region that cannot be fully covered by any $c-1$ points in $D$.

These results demonstrate that finite sets of source points cannot always guarantee complete visibility in general polygons.

The assumptions of Bonnet-Miltzow \cite{BonnetMiltzow2020}] (BM)

A simple polygon $\mathcal{P}$ satisfies:
\begin{enumerate}
    \item 
General position (Definition~\ref{def:general_position}).
 \item Integer coordinates: Vertices are positioned at integer coordinate points.
 \end{enumerate}

\subsection{Mildness of Assumption \ref{assumme:AR1}}

Our line aspect ratio assumption (Assumption~\ref{assumme:AR1}) encompasses a broad class of polygons, because it applies to:

\begin{enumerate}
    \item 
Polygons with real-coordinate vertices.
 \item  
Scenarios where vision stability does not necessarily hold.
\end{enumerate}

While Assumption~\ref{assume:AR2} (disk aspect ratio) fails to guarantee a finite set of source points $\mathcal{S}$ (Observation \ref{obs:ass2-infinite}), Assumption~\ref{assumme:AR1} (line aspect ratio) enables finite solutions. Lemma 4 and Lemma 5 from \cite{cite:EdouardTilmann} show that finite point sets cannot cover arbitrary visibility regions. However, our aspect ratio constraint overcomes these limitations, as formalized in Theorem \ref{thm:AR-generalizes-BM}.

\subsubsection{Complementarity of Line Aspect Ratio and Bonnet-Miltzow Frameworks}

The line aspect ratio (AR) and Bonnet-Miltzow (BM) frameworks complement each other:

\begin{enumerate}
    \item 
 There exist polygons that satisfy BM assumptions but have unbounded AR (Fig. \ref{fig:Bm-NotBM}).
    \item 
    There exist polygons with bounded AR that violate BM assumptions (e.g., polygons with bounded-AR real-coordinate vertices).
\end{enumerate}

\begin{figure}[tp]
\centering
\includegraphics[scale=0.9]{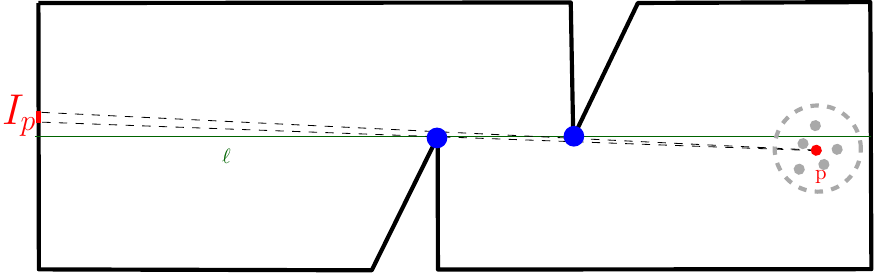}
\caption{A polygon satisfying BM's assumptions (general position and integer coordinates) but not satisfying line aspect ratio.}
\label{fig:Bm-NotBM}
\end{figure}

For consistency, we use $\text{AR} = \text{AR}_{\text{line}}$ to denote the line aspect ratio in subsequent discussions.

\section{Our Contribution}
\label{sec:contribution}
\begin{theorem}
\label{thm:finite-S}
Under the line aspect ratio assumption (Definition~\ref{assumme:AR1}), there exists a finite set $\mathcal{S}$ on $\overline{xy}$ such that:

$$|\mathcal{S}| \text{ is bounded by } \text{AR}$$
$$
\overline{uv} \subseteq \bigcup_{s \in \mathcal{S}} \text{VP}(s)
$$
\end{theorem}

\begin{proof}
The size of $\mathcal{S}$ is determined by $\text{AR} = LW/SW$:
\begin{enumerate}
    \item The slicing approach (Section~\ref{subsec:slicing-algorithm}) decomposes $\overline{uv}$ into visibility intervals.
    \item By Lemma \ref{lem:covering-target}, the points in $\S$ cover the target.
    \item Observation~\ref{obs:I-larger-SW} establishes that each interval has length $\geq SW$.
    \item Since $\overline{uv}$ has maximum length $LW$, the number of intervals is $\leq \frac{LW}{SW} = \text{AR}$.
\end{enumerate}
Thus $|\mathcal{S}|$ is finite and bounded by $\text{AR}$. 
\end{proof}

In the continue we present the slicing algorithm in Subsection \ref{subsec:slicing-algorithm}, and Subsection \ref{subsec:lemmas} covers the lemmas and observations and their proofs. Section \ref{sec:discussion} provides a final discussion.

\subsection{Slicing Algorithm}
\label{subsec:slicing-algorithm}
This subsection covers the slicing algorithm that splits a given source segment ($\seg{xy}$) by some middle points so that the union visibility of the set of all these points including the endpoints of the source segment covers an entire given target segment ($\seg{uv}$). Without lost of generality, we already suppose that the given source and target segments are weakly visible.

We start the slicing algorithm by defining two specific reflex vertices and their computing approach. 

Since the target is weakly visible from the source, consider the visibility of those points on the source whose view of the target is obstructed by some reflex vertices of $\P$. For each point on the source, its visibility can be blocked by at most two reflex vertices. However, these two reflex vertices may differ for different points on the source. For a precise definition of these reflex vertices, refer to Definition \ref{def:lbv-rbv}.

\begin{definition}
\label{def:lbv-rbv}
Consider two reflex vertices: $\LBV$, denoting the Left Blocking Vertex, and $\RBV$, representing the Right Blocking Vertex. These reflex vertices are defined with respect to a specific point on the source. For a point $q$ on the source, imagine standing at $q$, positioned between $x$ and $y$, while observing $\seg{uv}$. Assume that $x$ lies to the left and $y$ lies to the right of $q$. There exists a \textbf{single} reflex vertex on each side of $q$ such that $\LBV_{q}$ and $\RBV_{q}$ are uniquely determined by $q$ (see Observation~\ref{obs:lbvrbvunique}).

$\LBV_{q}$ (if it exists) is the reflex vertex where the line segment $\seg{q\LBV_{q}}$ intersects with $\seg{uv}$ and lies entirely inside $\P$, passes through at least one reflex vertex ($\LBV_{q}$), and has the exterior of $\P$ on the left side of $\seg{q\LBV_{q}}$.
If multiple reflex vertices lie on a single line crossing $\seg{q\LBV_{q}}$, the closest reflex vertex to $q$ along that line defines $\LBV_{q}$.

The same strategy defines $\RBV_{q}$, except that the exterior of $\P$ lies on the right side of $\seg{q\RBV_{q}}$.
\end{definition}

\subsubsection{Computing $\LBV_{q}$ and $\RBV_{q}$ for a point $q$ on $\seg{xy}$}:
\label{alg:computing-lbv-rbv}

We already know that $\seg{xy}$ and $\seg{uv}$ are weakly visible. Consider the line $\seg{qu}$ and run a sweeping algorithm on the reflex vertices of $\P$ to obtain a line that meets the requirements of Definition \ref{def:lbv-rbv}. That is a line that lies on at least one reflex vertex passing $\seg{uv}$ and holds other reflex vertices of $\P$ on its left side. Note that if multiple reflex vertices lie on this line, the closest reflex vertex to $q$ along that line defines $\LBV_{q}$.

Using the same sweeping algorithm on the other side with an opposite direction obtains $\RBV_{q}$.
\qed 
\subsubsection{Computing points on $\seg{xy}$}:
\label{alg:main}
Denoting $x$ as $x_{0}$ and $y$ as $y_{0}$, we will perform the iterations described below to compute a sequence of points $x_{i}$ and $y_{i}$ on $\seg{xy}$. This process continues until an iteration $j \geq 1$ is reached where $x_{j}$ lies to the right of $y_{j}$. We will demonstrate that, assuming $\P$ has a bounded line aspect ratio, the number of iterations has a polynomial upper bound (Lemma \ref{lem:finite-poly}). Furthermore, when $x_{j}$ lies to the right of $y_{j}$, the target will be covered by the set of points $\{x_{i}, y_{i} \mid 0 \leq i \leq j\}$ (denoted as $\mathcal{S}$) (see Lemma~\ref{lem:covering-target}).

\paragraph{Iterations of the algorithm after computing $\LBV$ and $\RBV$ vertices}:

 Consider two lines: the line intersecting $x_{i}$ and $\LBV_{x_{i}}$, and the line intersection $y_{i}$ and $\RBV_{Y_{i}}$, $i\geq 0$. 
 
 The line crossing $\seg{x_{i}\LBV{x_{i}}}$ intersect the target on $t_{x_{i}}$. The line crossing $\seg{y_{i}\RBV_{y_{i}}}$ intersects the target on $t_{y_{i}}$. 
 
 In each iteration $0 \leq i < j $, compute $t_{x_{i}}$ and $t_{y_{i}}$. 
 
 Consider $t_{x_{i}}$/$t_{y_{i}}$ as a middle point on $\seg{uv}$ where $v$ places at the left side of $t_{x_{i}}$. Assuming $\seg{uv}$ as the source compute $\LBV$ and $\RBV$ vertices for $t_{x_{i}}$ and $t_{y_{i}}$ points.
 
 Draw the line crossing $t_{x_{i}}$ and $\LBV_{t_{x_{i}}}$. The intersection of this line with $\seg{xy}$ creates a point denoted as $x_{i+1}$. 
 
  Draw the line crossing $t_{y_{i}}$ and $\RBV_{y_{i}}$. The intersection of this line with $\seg{xy}$ creates a point denoted as $y_{i+1}$.

If $x_{i+1}$ lies to the left of $y_{i+1}$, set $i = i+1$ and repeat the iteration procedure. Otherwise (that $x_{i+1}$ lies to the right of (or if they lie on one point) $y_{i+1}$), we reach the $j$th iteration and the slicing algorithm stops since the target is covered (Lemma \ref{lem:j-iteration} reveals that in the $j^{th}$ iteration the target gets covered successfully).

 In case one of the points $\LBV$, $\RBV$ does not exist, the corresponding lines do not exist as well. If it is an $x$ or $y$ points the point in that iteration can see the rest of the target. If it is a $t$ point it can see rest of the source so the next point on the next point on the source is $x$ or $y$ itself. So, the algorithm has already reached a position where the points can see the entire target and the slicing algorithm terminates. 
 
 \setlength{\textfloatsep}{0pt}
\setlength{\intextsep}{0pt}
 \begin{figure}[htp]
\begin{center}
 \includegraphics[scale=0.5]{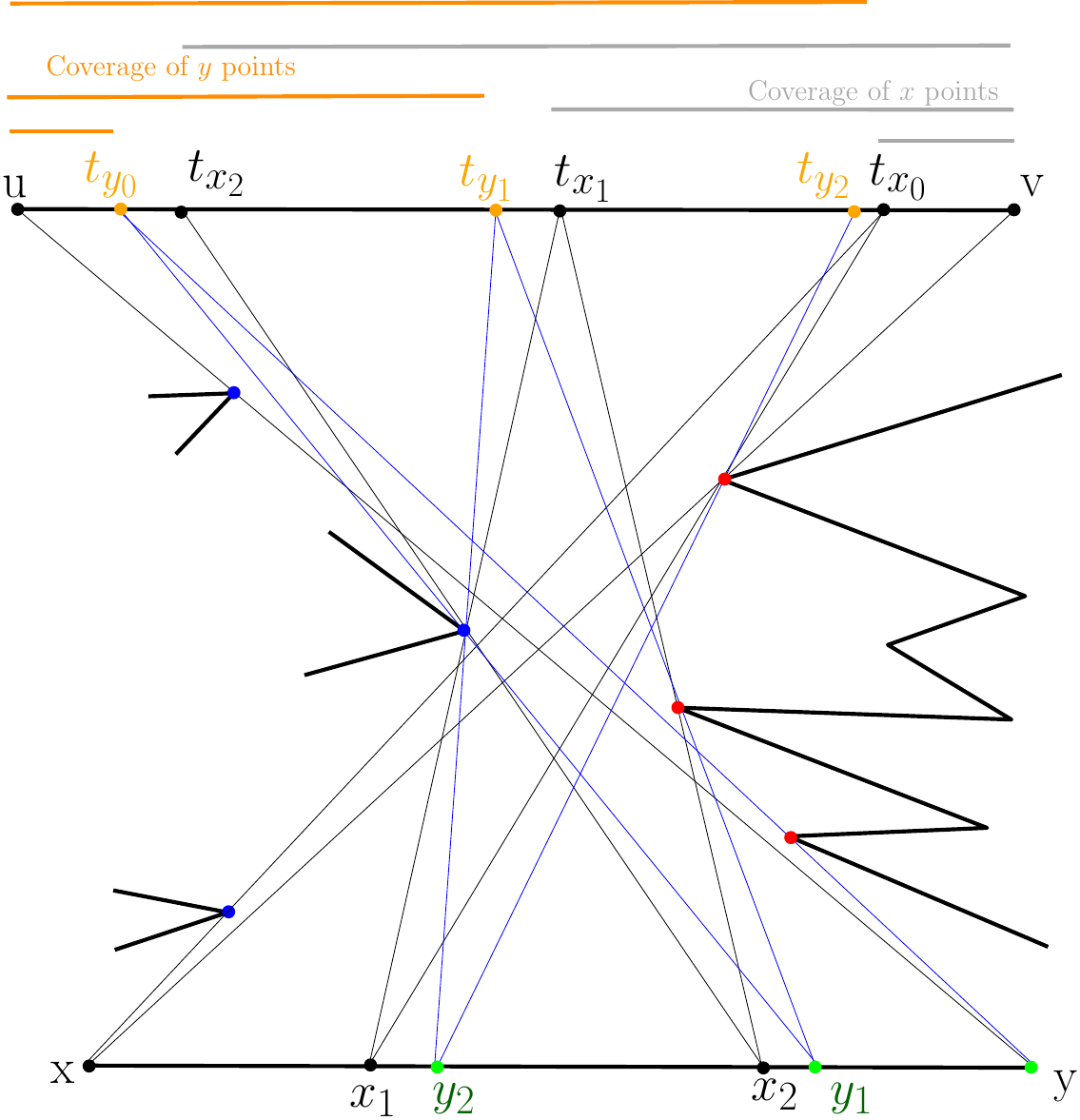}
\caption{
The coverage from both the left and right sides is demonstrated for each iteration. In the second iteration ($j=2$), the slicing algorithm concludes, having successfully covered the entire target segment.
 }
 \label{fig:j-iteration-coverage}
\end{center}
 \end{figure}

Figure \ref{fig:j-iteration-coverage} provides an example that illustrates the iterative process of the slicing algorithm.

\textit{End of the iteration.}

The set of all $x_{i}$ and $y_{i}$ points determines $\mathcal{S}$. 
 
\textbf{End of the slicing algorithm.}

\subsubsection{Results of the slicing algorithm}

Lemma \ref{lem:covering-target} indicates that the set $\S$ obtained by the slicing algorithm covers the target. Lemma \ref{lem:finite-poly} we know that under the cases of Assumption \ref{assumme:AR1} $|\S|$ remains polynomial in $n$.

\subsection{Observations, Lemmas, Theorems, and their proofs}
\label{subsec:lemmas}

\begin{obs}
\label{obs:ass2-infinite}
Given two segments a source $\seg{xy}$ and $\seg{uv}$ inside a simple polygon $\P$, Assumption \ref{assume:AR2} cannot guarantee a finite set $\S$ of points on the source to cover the target.
\end{obs}
\begin{proof}

See Fig. \ref{fig:sd-ld-sw}. Based on \textit{Lemma 4 of \cite{cite:EdouardTilmann}} mentioned previously, we cannot find an finite set of points around $p$ (including the sub-segment of the source around $p$) that the union visibility of the points in the set can cover the visibility of $p$. Fig. \ref{fig:sd-ld-sw} illustrates a counter example for Assumption \ref{assume:AR2}, where we can set the ratio of $\frac{LD}{SD}$ to be large enough without modifying the size of $SW$. Still the position of $p$ and $\ell$ can be set so that $p$ sees $I_{p}$ as an interval on the target. For enlarging the ratio $\frac{LD}{SD}$, we can enlarge the minimal circle by taking the reflex vertices away, in fact we can move the reflex vertices on the parallel lines and provide a large polygon without changing SW. So, Assumption \ref{assume:AR2} cannot guarantee of finding a finite set of points on the source to cover the target.

\setlength{\textfloatsep}{5pt}
\setlength{\intextsep}{5pt}
\begin{figure}[tp]
\begin{center}
\includegraphics[scale=0.9]{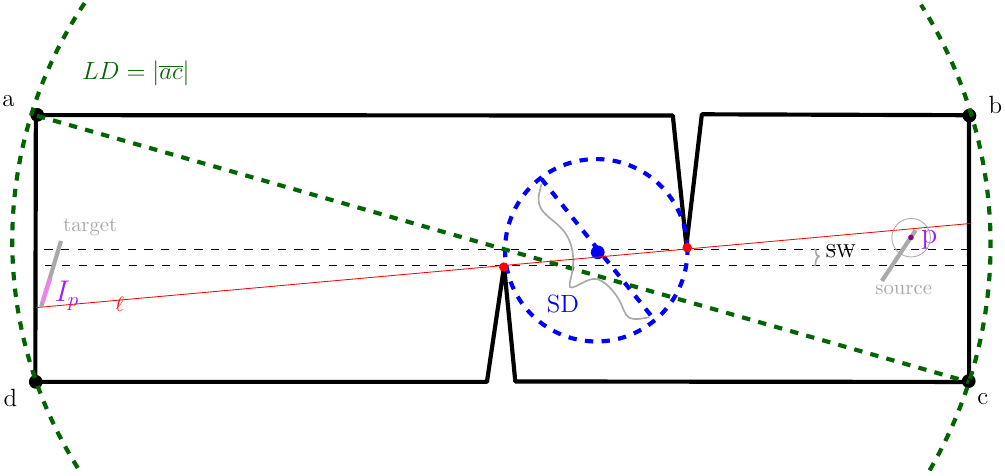}
\caption{SW based on Assumption \ref{assumme:AR1} is more effective that SD of Assumption \ref{assume:AR2}.}
\label{fig:sd-ld-sw}
\end{center}
\end{figure}
\end{proof}

\begin{obs}
\label{obs:lbvrbvunique}
  Given a segment $\overline{xy}$ as the source and a target segment $\seg{uv}$ that is partially visible to $\overline{xy}$, consider a midpoint on the source, denoted as $q$. There exists a unique reflex vertex obstructing the visibility of $q$, denoted by $\LBV_{q}$. This vertex is unique for $q$. Specifically, if we stand on $\overline{xy}$ with $x$ to our left, $\LBV$ blocks the visibility of $q$ from the left side (if it exists). Similarly, $\RBV$ is unique (if it exists) and blocks the visibility of $q$ from the right side of $\overline{xy}$.
\end{obs}
\begin{proof}
We have to prove that $\LBV_{q}$ and $\RBV_{q}$ are unique reflex vertices on each side for a specific point $q$ on $\seg{xy}$. Suppose considering the condition of the lemma both of these reflex vertices exist.

 Without lost of generality consider $\LBV_{q}$. On the contrary, suppose it is not unique. For proof, let's consider another reflex vertex denoted by $lrv \neq \LBV_{q}$, which could potentially obstruct the visibility of $q$ (a point on $\seg{xy}$) not to see some part of the target from the left side.

    See Fig. \ref{fig:lbvrbv}. To begin, we show that the visibility of $q$ cannot be obstructed by any other reflex vertex aside from $\LBV_{q}$.
    Suppose, for the sake of contradiction, that there exists a reflex vertex $lrv$ on the left side of $\overline{q\LBV_{q}}$, which obstructs the visibility of $q$, preventing it from seeing a portion of the target.
    In such a scenario, the line crossing $\overline{q lrv}$ should holds $\LBV$ on its left side. Otherwise, $lrv$ defines $\LBV_{q}$ itself.

    The same analysis reveals that $\RBV$ is unique for $q$ on the other side.

\setlength{\textfloatsep}{0pt}
\setlength{\intextsep}{0pt}
\begin{figure}[htp]
\begin{center}
\includegraphics[scale=0.6]{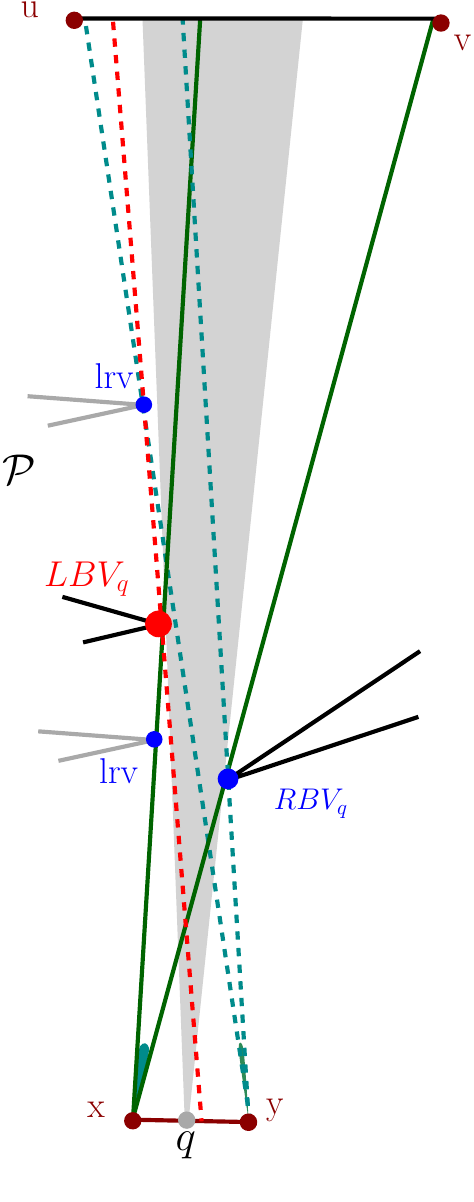}
\caption{
Considering a source and a target and a point $q$ on $\seg{xy}$, $\LBV_{q}$ and $\RBV_{q}$ vertices (if exist) are unique.  
}
\label{fig:lbvrbv}
\end{center}
\end{figure}

\end{proof}

\begin{obs}
\label{obs:I-larger-SW}
Any point $p$ on the source sees an interval $I_{p}$ which $I_{p} < SW$, where $SW$ comes from Assumption \ref{assumme:AR1}. 
\end{obs}
\begin{proof}
A point $p$ moving from $x$ to $y$ on the source, the half-lines through $\seg{p\LBV}$ and $\seg{p\RBV}$ are divergent. So, the interval created on the target are larger than distance between the parallel lines crossing $\LBV$ and $\RBV$. 
\end{proof}

\begin{lemma}
\label{lem:covering-target}
Given two weakly visible segments, $\seg{xy}$ (the source) and $\seg{uv}$ (the target), inside a simple polygon $\P$, the slicing algorithm described in Subsection \ref{subsec:slicing-algorithm} generates a set of points on $\seg{xy}$, denoted by $\mathcal{S}$, that collectively cover the entire target.
\end{lemma}

\begin{proof}

Without loss of generality, consider $x_i$. Any point $x_i$ sees the target between two lines: one crossing $\seg{x_i\LBV_{x_i}}$ and the other crossing $\seg{x_i\RBV_{x_i}}$.
Note that the target is weakly visible to the source. Thus, $x = x_0$ must see $v$, and the visibility of the $x$ points progressively aims to cover the target from $v$ to $u$. The reflex vertices that block the visibility of the $x$ series from seeing a part of the target are the $\LBV$ vertices.
In each iteration $i$, the line crossing $\seg{x_i\LBV_{x_i}}$ determines $t_{x_{i+1}}$ (see Fig. \ref{fig:target-sequence}), and $x_{i+1}$ can see the target starting from $t_{x_{i+1}}$. Therefore, the visibility of the $x_i$ points on the target are connected.
Thus, the target is visible to $x_i$ from $t_{x_i}$ to $t_{x_{i+1}}$. This process continues until the iteration stops, either when $x_i$ sees $u$ or when reaching a point $y_i$ where the remaining portion of the target has already been covered by $y$ points from the previous iterations.

\setlength{\textfloatsep}{0pt}
\setlength{\intextsep}{0pt}
\begin{figure}[htp]
\begin{center}
\includegraphics[scale=0.52]{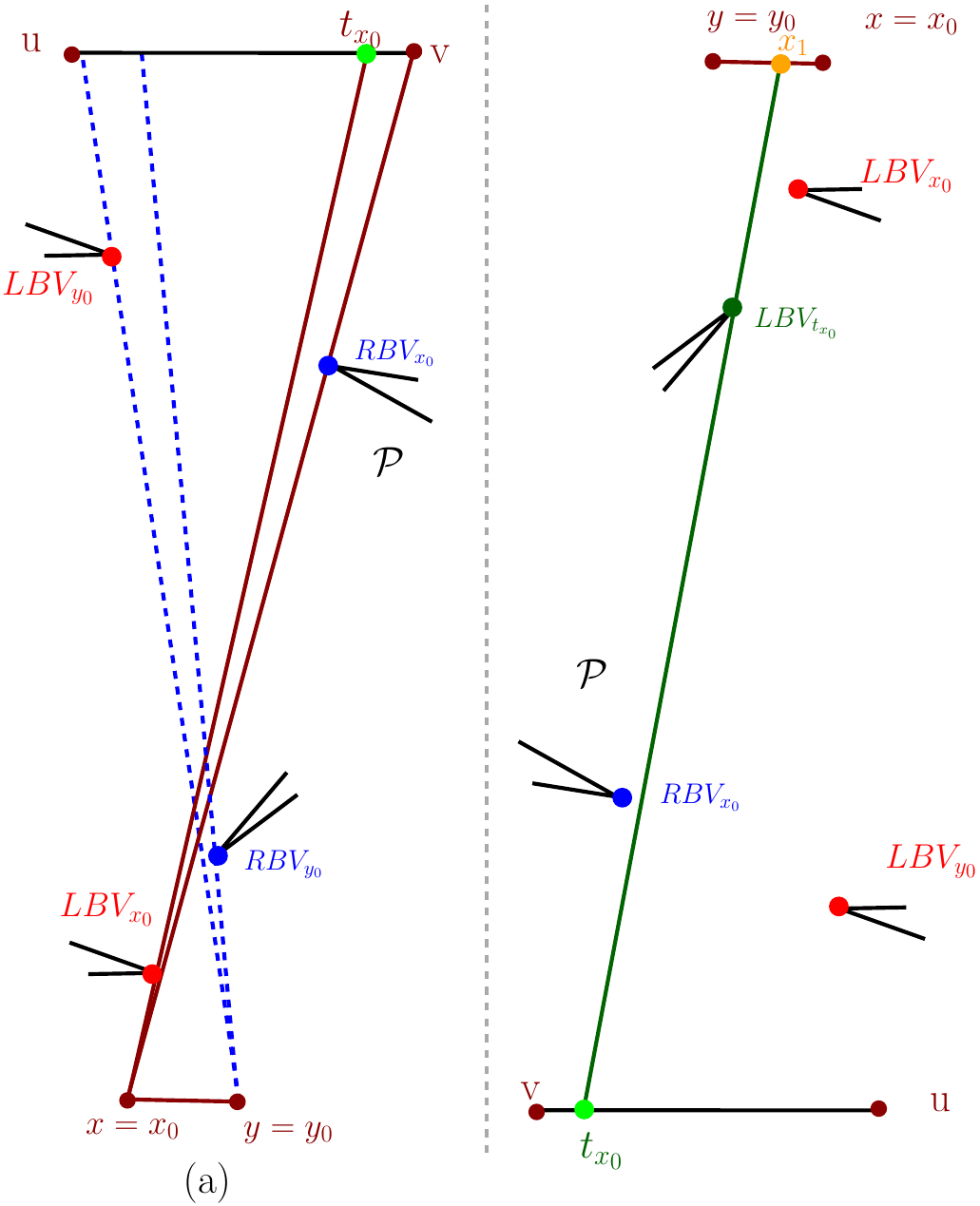}
\caption{
The first iteration of the slicing algorithm. 
}
\label{fig:target-sequence}
\end{center}
\end{figure}

\end{proof}

\begin{lemma}
\label{lem:finite-poly}
The output of the slicing algorithm (as detailed in Subsection \ref{subsec:slicing-algorithm}) produces a finite set of points, $\mathcal{S}$, on the source. Under Assumption \ref{assumme:AR1}, the size of $\mathcal{S}$ is polynomial in $n$, where $n$ represents the complexity of $\P$.
\end{lemma}
\begin{proof}
Without loss of generality, we present the proof considering only the points labeled $x$.

First, observe that $\RBV_{x_{i+1}}$ and $\LBV_{t_{x_i}}$ must be the same reflex vertex. If they were different, $x_{i+1}$ would be able to see a point closer to $v$ on the target, and $\LBV_{t_{x_i}}$ would not be obstructing its visibility.

Thus, in each iteration, $x_i$ sees the target between two lines: one crossing $\seg{x_i\LBV_{x_i}}$ and the other crossing $\seg{x_i\RBV_{x_i}}$. From Observation \ref{obs:I-larger-SW}, we know that the number of points on the source obtained by the slicing algorithm is upper bounded by the ratio of $LW$ to $SW$.
\end{proof}

\begin{lemma}
\label{lem:j-iteration}
Given two weakly visible segments, the source $\seg{xy}$ and the target $\seg{uv}$, inside a simple polygon $\P$, the slicing algorithm presented in Subsection \ref{subsec:slicing-algorithm} completes its execution in the $j^{th}$ iteration, where $y_j$ lies to the left of $x_j$. In this iteration, the target segment is guaranteed to be fully covered.
\end{lemma}
\begin{proof}
We only have to prove that $t_{y_{j}}$ lies to the left of (or on) $t_{x_{i}}$. This means that the coverage of the target from the left meets the coverage from the right side.

Without loss of generality, suppose the visibility of both $x_j$ and $y_j$ is obstructed by reflex vertices. Consider the triangle formed by the points $x_j$, $y_j$, and $\LBV_{x_j}$. Since $y_j$ is to the left of $x_j$ and has visibility of some part of the target, $y_j$ must lie to the left of the line passing through $\seg{x_j\LBV_{x_j}}$. Similarly, analyzing the triangle formed by $x_j$, $y_j$, and $\RBV_{y_j}$ reveals that $x_j$ lies to the right of the line passing through $\seg{y_j\RBV_{y_j}}$.
Given that $\LBV_{x_j}$ lies to the left of the interior of $\P$ and $\RBV_{y_j}$ lies to the right of the interior of $\P$, the lines passing through $\seg{x_j\LBV_{x_j}}$ and $\seg{y_j\RBV_{y_j}}$ must intersect, placing their intersection on the target. Consequently, $t_{x_j}$ lies to the left of $t_{y_i}$ on the target.

 \setlength{\textfloatsep}{0pt}
\setlength{\intextsep}{0pt}
 \begin{figure}[htp]
\begin{center}
 \includegraphics[scale=0.4]{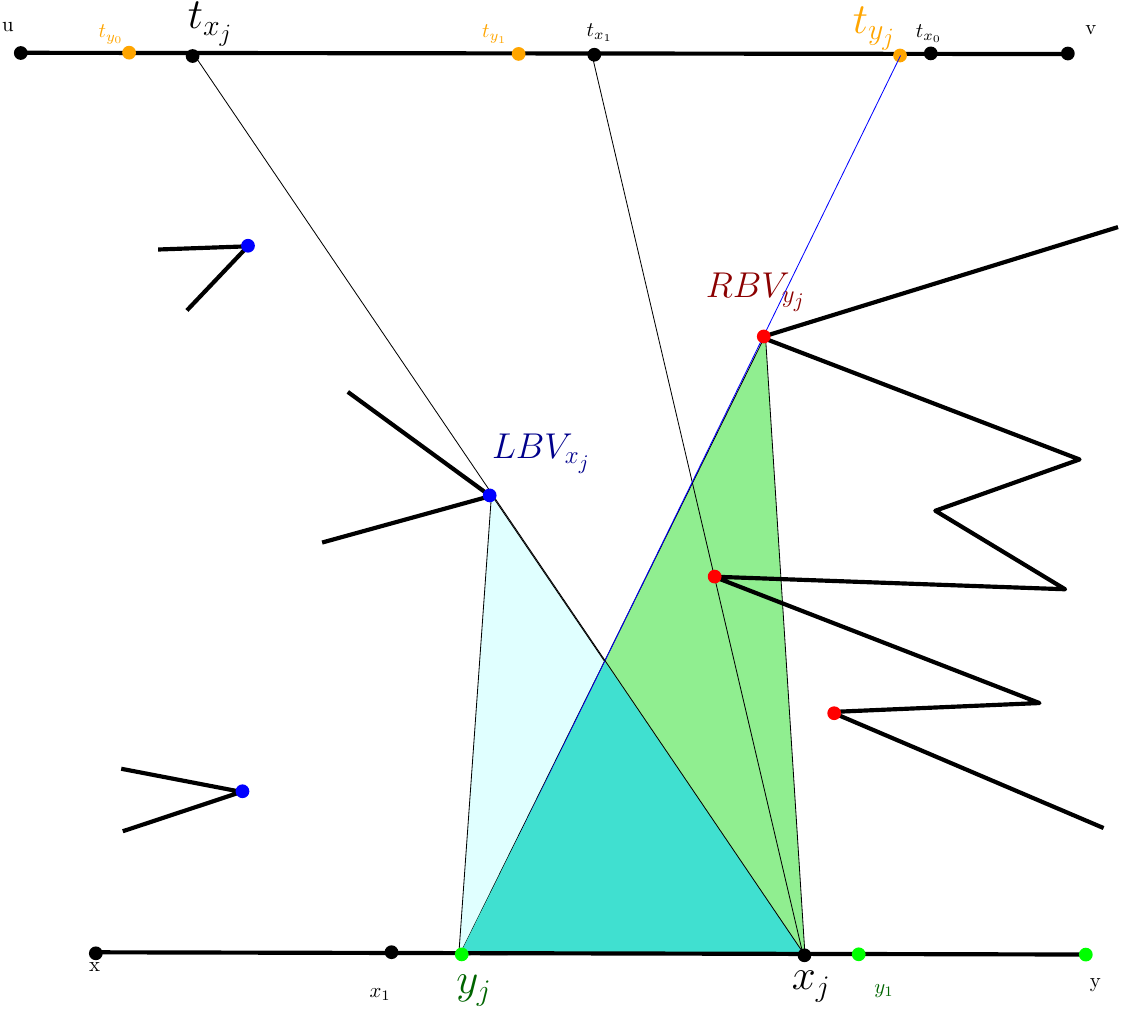}
\caption{
In this example, $j=2$. It demonstrates the $j^{th}$ iteration, where $y_j$ is positioned to the left of $x_j$. 
 }
 \label{fig:j-iteration-coverage}
\end{center}
 \end{figure}
 
In case $x_j$ and $y_j$ coincide at a single point, the $\LBV$ and $\RBV$ of that point are distinct. The lines passing through $x_j$ (or $y_j$) and $\LBV$, and through $x_j$ (or $y_j$) and $\RBV$, intersect the target at different points such that $t_{x_j}$ lies to the left of $t_{y_j}$.
\end{proof}

\begin{theorem}[Line Aspect Ratio Generalizes Vision Stability]
\label{thm:AR-generalizes-BM}
For any simple polygon $\mathcal{P}$ satisfying Bonnet and Miltzow's assumptions:
 General position (no three vertices collinear), and
   Vision stability with constant $\gamma(\mathcal{P})$.
 $\mathcal{P}$ has bounded line aspect ratio: $\text{AR}_{\text{line}} = \frac{LW}{SW} = O(1)$. 
\end{theorem}

\begin{proof}
\textbf{Vision Stability $\Rightarrow$ Bounded AR}\\ 
Let $\mathcal{P}$ satisfy vision stability with constant $\gamma$. Consider parallel support lines $L_1, L_2$ tangent to reflex vertices $v_i, v_j$ at distance $d = SW$. Place guards $p \in L_1 \cap \partial\mathcal{P}$, $q \in L_2 \cap \partial\mathcal{P}$. 

By vision stability:
\begin{equation}
\mu(\mathcal{V}(p) \triangle \mathcal{V}(q)) \leq \gamma \|p - q\|
\end{equation}
The strip between $L_1$ and $L_2$ contains a visibility corridor of area $\Omega(d^2)$ (Lemma~\ref{lem:strip-area}). Thus:
\begin{align}
c_1 d^2 &\leq \gamma d \\
d &\geq \frac{c_1}{\gamma} > 0
\end{align}
where $c_1 > 0$ is a geometric constant. Since $LW \leq \text{diam}(\mathcal{P}) \leq c_2 \cdot \text{perim}(\mathcal{P})$ for fixed $c_2$:
\begin{equation}
\text{AR}_{\text{line}} = \frac{LW}{SW} \leq \frac{c_2 \cdot \text{perim}(\mathcal{P}) \cdot \gamma}{c_1} = O(1)
\end{equation}

\end{proof}

\begin{lemma}
\label{lem:strip-area}
For two parallel lines $L_{1}$ and $L_{2}$ at distance $d$ tangent to reflex vertices in a simple polygon $\P$, the strip between them contains a region of area $\Omega(d^{2})$.  
\end{lemma}

\begin{proof}
Let $L_1$ and $L_2$ be two parallel lines at distance $d$, tangent to reflex vertices of a simple polygon $\mathcal{P}$. Since $L_1$ and $L_2$ are tangent, the polygon must touch both lines, ensuring that $\mathcal{P}$ spans the strip between them.

The region of $\mathcal{P}$ within this strip must occupy an area dictated by the strip’s width $d$. The minimal configuration occurs when $\mathcal{P}$ forms a parallelogram spanning the strip, with base and height both equal to $d$, yielding an area of $d^2$. In degenerate cases, the polygon may form a triangular region within the strip, with area $\frac{d^2}{2}$.

Since both cases satisfy the lower bound of $\Omega(d^2)$, the result follows.

To be more precise:
\begin{enumerate}
\item
 The tangent lines create a ``corridor'' of width $d$.  
  \item  By the isoperimetric inequality, the minimal-area shape fitting this corridor is a rectangle or parallelogram.  
  \item  The polygon must occupy at least the area of a parallelogram with:  
   Base = $d$, and Height = $d$ (ensuring area $d^{2}$)  
  \item  In degenerate cases, the region may be a triangle with area $d^{2}/2$, preserving $\Omega(d^{2})$.  
  
\end{enumerate}

\end{proof}

\section{Discussion}
\label{sec:discussion}
Our work addresses a core problem in segment guarding by ensuring finite solutions without relying on restrictive stability assumptions. We introduce the novel concept of the \textit{line aspect ratio} (AR), a geometric parameter quantifying the anisotropy of reflex vertices. This framework guarantees finite guard sets $(\mathcal{S})$ with a size bounded by $\mathcal{O}(\text{AR})$. Our approach accommodates all vision-stable polygons while extending to broader classes of polygons.

\subsection{A few open problems}
\begin{enumerate}
    \item 
 Optimality of $|\mathcal{S}|$: Is \(|\mathcal{S}| = \Theta(\text{AR})\) optimal?  
 \item 
Dynamic Guarding: Can the guard set $\mathcal{S}$ adapt to polygon deformation while maintaining bounded size?
 \item 
 Extension to 3D: How can the AR framework be generalized for guarding problems in three-dimensional environments?
\end{enumerate}

The AR assumption aligns with realistic geometric constraints observed in various domains:
\begin{enumerate}
    \item 
\textit{Architectural Layouts}: Reflex vertices often form corridors with bounded anisotropy, such as in building floor plans \cite{schwartzburg2014}.
\item \textit{Geographic Meshing}: Terrain models exhibit moderate variations in width and feature distribution \cite{puppo1997}. 
\item \textit{Sensor Networks}: Efficient sensor deployment frequently leverages regions with bounded aspect ratios \cite{wang2007}.
\item \textit{Robotic Navigation}: Structured environments simplify visibility reasoning for autonomous systems \cite{lavalle2004}.
\end{enumerate}

This framework bridges theoretical advances and practical applications, providing a robust foundation for future research in visibility and guarding problems

\bibliographystyle{IEEEtran}
\bibliography{IEEEfull,lipics-v2021-sample-article}

\begin{thebibliography}{10}
\providecommand{\url}[1]{#1}
\csname url@samestyle\endcsname
\providecommand{\newblock}{\relax}
\providecommand{\bibinfo}[2]{#2}
\providecommand{\BIBentrySTDinterwordspacing}{\spaceskip=0pt\relax}
\providecommand{\BIBentryALTinterwordstretchfactor}{4}
\providecommand{\BIBentryALTinterwordspacing}{\spaceskip=\fontdimen2\font plus
\BIBentryALTinterwordstretchfactor\fontdimen3\font minus
  \fontdimen4\font\relax}
\providecommand{\BIBforeignlanguage}[2]{{%
\expandafter\ifx\csname l@#1\endcsname\relax
\typeout{** WARNING: IEEEtran.bst: No hyphenation pattern has been}%
\typeout{** loaded for the language `#1'. Using the pattern for}%
\typeout{** the default language instead.}%
\else
\language=\csname l@#1\endcsname
\fi
#2}}
\providecommand{\BIBdecl}{\relax}
\BIBdecl

\bibitem{avis}
D.~Avis and G.~T. Toussaint., ``An optional algorithm for determining the
  visibility of a polygon from an edge.'' \emph{IEEE Transactions on
  Computers}, no. C-30, pp. 910--1014, 1981.

\bibitem{2}
L.~J. Guibas, J.~Hershberger, D.~Leven, M.~Sharir, and R.~E. Tarjan.,
  ``Linear-time algorithms for visibility and shortest path problems inside
  triangulated simple polygons.'' \emph{Algorithmica}, no.~2, pp. 209--233,
  1987.

\bibitem{avis1981optimal}
D.~Avis and G.~T. Toussaint, ``An optimal algorithm for determining the
  visibility of a polygon from an edge,'' \emph{IEEE Transactions on
  Computers}, vol.~30, no.~12, pp. 910--914, 1981.

\bibitem{guibas1987visibility}
L.~J. Guibas, R.~Motwani, and P.~Raghavan, ``Visibility of disjoint polygons,''
  \emph{Algorithmica}, vol.~1, pp. 49--63, 1987.

\bibitem{hershberger1989optimal}
J.~Hershberger and S.~Suri, ``An optimal algorithm for euclidean shortest paths
  in the plane,'' \emph{SIAM Journal on Computing}, vol.~28, no.~6, pp.
  2215--2256, 1989.

\bibitem{sack1991weak}
J.-R. Sack and S.~Suri, ``Weak visibility of two simple polygons,''
  \emph{Computational Geometry: Theory and Applications}, vol.~1, pp. 43--58,
  1991.

\bibitem{ghodsa2008weak}
M.~Ghods and A.~Mirzaian, ``Weak visibility queries between disjoint line
  segments in polygons,'' \emph{Computational Geometry: Theory and
  Applications}, vol.~39, no.~4, pp. 239--249, 2008.

\bibitem{toth2010art}
C.~D. T{\'o}th, ``Art gallery problem with guards on edges,'' \emph{SIAM
  Journal on Discrete Mathematics}, vol.~24, no.~1, pp. 1--11, 2010.

\bibitem{king2004guarding}
V.~King and J.~Snoeyink, ``Guarding polygons with mobile guards,''
  \emph{Computational Geometry: Theory and Applications}, vol.~26, no.~3, pp.
  209--219, 2004.

\bibitem{bhattacharya2015approximability}
B.~Bhattacharya and S.~K. Ghosh, ``Approximability of guarding weakly visible
  polygons,'' in \emph{Proceedings of the 31st International Symposium on
  Computational Geometry (SoCG)}.\hskip 1em plus 0.5em minus 0.4em\relax
  Springer, 2015, pp. 201--215.

\bibitem{biedl2017sliding}
T.~C. Biedl, T.~M. Chan, and F.~Montecchiani, ``Sliding cameras in orthogonal
  polygons: Approximation, hardness, and art gallery bounds,'' in
  \emph{Proceedings of the Symposium on Computational Geometry (SoCG)}, 2017,
  pp. 45:1--45:16.

\bibitem{vaeziwalcom}
A.~Vaezi and M.~Ghodsi, ``How to extend visibility polygons by mirrors to cover
  invisible segments,'' in \emph{WALCOM: Algorithms and Computation}, S.-H.
  Poon, M.~S. Rahman, and H.-C. Yen, Eds.\hskip 1em plus 0.5em minus
  0.4em\relax Cham: Springer International Publishing, 2017, pp. 42--53.

\bibitem{tcs}
------, ``Visibility extension via reflection-edges to cover invisible
  segments.'' \emph{Theoretical Computer Science}, 2019.

\bibitem{lee1990chain}
D.~T. Lee and K.-Y. Chwa, ``Visibility of a simple chain in a polygon,'' in
  \emph{Proceedings of the 6th Annual Symposium on Computational Geometry
  (SoCG)}, 1990, pp. 211--220.

\bibitem{algowin2023transmitters}
M.~Authors, ``k-transmitters for segment visibility: Algorithms and
  applications,'' \emph{ALGOWIN}, vol.~10, pp. 100--115, 2023.

\bibitem{chen1998structured}
D.~Chen, ``Structured visibility profiles in polygons,'' in \emph{Proceedings
  of the Symposium on Computational Geometry (SoCG)}, 1998, pp. 23--30.

\bibitem{bonnet2016approximating}
{\'E}.~Bonnet and T.~Miltzow, ``Approximating the art gallery problem,''
  \emph{Leibniz International Proceedings in Informatics (LIPIcs)}, vol.~51,
  pp. 17:1--17:15, 2016.

\bibitem{BonnetMiltzow2016}
------, ``Approximating the art gallery problem,'' in \emph{Proceedings of the
  32nd International Symposium on Computational Geometry (SoCG)}, ser. Leibniz
  International Proceedings in Informatics (LIPIcs), vol.~51.\hskip 1em plus
  0.5em minus 0.4em\relax Schloss Dagstuhl--Leibniz-Zentrum f{\"u}r Informatik,
  2016, pp. 17:1--17:15.

\bibitem{cite:EdouardTilmann}
\BIBentryALTinterwordspacing
E.~Bonnet and T.~Miltzow, ``An approximation algorithm for the art gallery
  problem,'' \emph{The 33rd International Symposium on Computational Geometry
  (SoCG'17)}, vol.~77, pp. 20:1--20:15, 2017. [Online]. Available:
  \url{https://hal.science/hal-01994349v1}
\BIBentrySTDinterwordspacing

\bibitem{BonnetMiltzow2020}
\BIBentryALTinterwordspacing
{\'E}.~Bonnet and T.~Miltzow, ``An approximation algorithm for the art gallery
  problem,'' \emph{Algorithmica}, vol.~82, pp. 630--673, 2020. [Online].
  Available: \url{https://doi.org/10.1007/s00453-019-00667-5}
\BIBentrySTDinterwordspacing

\bibitem{schwartzburg2014}
Y.~Schwartzburg and M.~Pauly, ``High-resolution topological tools for immersive
  architectural design,'' \emph{Computer-Aided Design}, vol.~55, pp. 45--57,
  2014, demonstrates bounded aspect ratios in architectural feature
  decomposition.

\bibitem{puppo1997}
E.~Puppo, ``Variable resolution terrain surfaces,'' \emph{Proc. CG
  International}, pp. 81--90, 1997, shows natural terrains have bounded width
  ratios in mesh simplification.

\bibitem{wang2007}
B.~Wang and K.~C. Chua, ``Coverage in hybrid mobile sensor networks,'' in
  \emph{IEEE MASS}, 2007, pp. 1--8, uses aspect-ratio constraints for sensor
  placement optimization.

\bibitem{lavalle2004}
S.~M. Lavalle, ``Probabilistic roadmaps for path planning in high-dimensional
  configuration spaces,'' \emph{IEEE Transactions on Robotics}, vol.~12, no.~4,
  pp. 566--580, 2004, assumes bounded environment anisotropy for efficient
  visibility sampling.

\end{thebibliography}

\end{document}